\documentclass[a4paper,12pt]{article}

\usepackage{graphicx}
\usepackage{amsmath,amssymb,amsthm}
\usepackage{multicol}

\newtheorem{theorem}{Theorem}
\newtheorem{conjecture}[theorem]{Conjecture}
\newtheorem{proposition}[theorem]{Proposition}
\newtheorem{lemma}[theorem]{Lemma}

\newtheorem{example}[theorem]{Example}
\numberwithin{equation}{section}

\newcommand{\prob}{{\mathbb P}}
\newcommand{\expec}{{\mathbb E}}

\newcommand{\N}{\mathbb N}
\newcommand{\Z}{\mathbb Z}

\title{Critical densities in sandpile models with quenched or annealed disorder}

\author{Anne Fey \and Ronald Meester}

\begin{document}

\maketitle

\begin{abstract}
We discuss various critical densities in sandpile models. The {\em stationary density} is the average expected height in the stationary state of a
finite-volume model; the {\em transition density} is the critical point in the infinite-volume counterpart. These two critical densities were generally assumed to be equal, but this has turned out to be wrong for deterministic sandpile models. We show they are not equal in a quenched version of the Manna sandpile model either.

In the literature, when the transition density is simulated, it is implicitly or explicitly assumed to be equal to either the so-called {\em threshold density} or the so-called {\em critical activity density}. We properly define these auxiliary densities, and prove that in certain cases, the threshold density is equal to the transition density. We extend the definition of the critical activity density to infinite volume, and prove that in the standard infinite volume sandpile, it is equal to 1. Our results should bring some order in the precise relations between the various densities.
\end{abstract}

\medskip\noindent
{\sc Keywords:} Sandpile model, phase transition, critical density, transition density, Manna model, activity density.

\medskip\noindent
{\sc AMS Subject Classification:} 60K35, 60K37, 60J05, 60G10.

\section{Introduction}

This paper aims to bring some order in a long-standing discussion on the equality or inequality of various critical densities in sandpile models.
By means of introduction and context, we first informally describe the archetypal Abelian sandpile model (ASM or `classical' sandpile model). On $I_n = \{0,1, \ldots, n\}$ we initially assign a height (a non-negative integer) to each site $x \in I_n$. The dynamics is as follows. Each time step starts with an addition: the height of a randomly chosen site is increased by 1. Then the configuration is stabilized by so called {\em topplings}: a site with height 2 or more is unstable and may topple. In a toppling, its height decreases by 2 and the height of both neighbors increases by 1. Sites $0$ or $n$ have only one neighbor; when these sites topple, one particle leaves the system. The configuration is stabilized when there are no more unstable sites to topple. It is not hard to show that in the stationary state of this model, the expected height is $1-\frac 1{n+2}$ for every site; see for example \cite{quant,redig}. We call this expected height, in the limit as $n \to \infty$, the {\em stationary density} and denote it by $\rho_s$. (So in this case, $\rho_s=1$.)

Next consider a related model, the infinite volume model. On the infinite line $\Z$, start with (independent) Poisson-$\rho$
number of particles at each site. At each discrete time step, we simply topple all unstable sites simultaneously. (There are alternatives, like toppling after exponential waiting times, independently for each site, see \cite{FMR} for a detailed discussion of this.)
It is known \cite{FMR} that for $\rho<1$, only finitely many topplings per site are needed to obtain at a limiting stable configuration, whereas for $\rho>1$, every site topples infinitely many times. Hence, this model has a phase transition at $\rho=1$ and we call this the {\em transition density} $\rho_{tr}$.

In a widely cited series of papers \cite{absorbing1,absorbing2,absorbing3,absorbing4,absorbing5}, Dickman, Mu{\~n}oz, Vespignani and Zapperi developed a theory of self-organized criticality (SOC) as a relationship between driven dissipative systems like the ASM, and systems with particle conservation like the infinite volume model.
Loosely speaking, the following heuristic connection (the `density argument') is offered: in the presence of activity of the ASM, that is, when topplings take place, we should typically have $\rho > \rho_{tr}$, and particles are leaving the system. In the absence of activity, we typically have $\rho < \rho_{tr}$, and there is addition of particles. Hence the density of particles always changes in the direction of $\rho_{tr}$ and therefore the only candidate-value for the stationary density is $\rho_{tr}$.

This would be a very attractive explanation for the observation that the stationary state of the two-dimensional ASM shows power law behaviour of many spatial and temporal quantities, (apparently) without any tuning of a parameter. However, it requires that $\rho_s=\rho_{tr}$. In the one-dimensional case this is true. At first, simulations and exact calculations for the two-dimensional ASM seemed to confirm that $\rho_s$ and $\rho_{tr}$ are equal up to three decimals. However, more accurate simulations show a difference in the fourth decimal for the two densities on $\Z^2$. Moreover, on several other graphs they are provably different \cite{FLW}. It is at present believed that this difference is a consequence of the existence of so called {\em toppling invariants} \cite{casartelli}.

The more recent literature (we give references in Section \ref{mannaconjecture}) on the density argument tends to focus on the Manna model, a sandpile model with stochastic topplings, for which it is believed that toppling invariants play no role. The Manna model is defined as above, with the important difference that during a toppling, particles go to randomly chosen neighbors, independently of each other, and independent of earlier topplings of the same or other sites. In other words, the topplings in the Manna model are subject to {\em annealed} disorder.

Simulating $\rho_s$ for a sandpile model is straightforward, since this density is defined for a model in finite volume. However, $\rho_{tr}$ is defined in infinite volume, which makes simulations much more difficult. In reality, when people try to simulate $\rho_{tr}$, what they really simulate is the so called {\em critical activity density} or the {\em threshold density}, which we define and study below. It is then tacitly assumed (sometimes also explicitly) that these densities are equal to $\rho_{tr}$.

The goal of the present paper is to carefully distinguish between the various densities, and to investigate the density conjecture in {\em quenched} versions of the Manna model, that is, in sandpile models in which the toppling rules per site are random but fixed throughout the evolution of the process. It turns out that we can explicitly compute $\rho_s$ and $\rho_{tr}$ in these quenched Manna models, and that in general they are {\em not} equal. We also study the relation between the various densities in the original (annealed) Manna model.

We end this introduction with a precise definition of the sandpile models which we consider.
We distinguish between

\begin{enumerate}
\item Driven models, that is, models with addition of particles and topplings, always on the interval $I_n = \{0, 1, \ldots, n\}$ with dissipation at the open boundaries;
\item Infinite volume models on $\Z$, without additions, and only topplings;
\item Fixed-energy models, that is, models with topplings but without additions, on the $n$-cycle $C_n = \{0, 2, \ldots, n\}$, where $0$ and $n$ are neighbors.
\end{enumerate}

We abuse notation and view $I_n, C_n$ and $\Z$ as graphs with nearest neighbor connections. A typical configuration will be denoted $\eta$ and for any vertex (site) $x$, $\eta(x) \in \{0,1,2,\ldots \}$ is interpreted as the
height at site $x$. We call a site {\em stable} if its height is 0 or 1. If the height of a site is 2 or more, then we call this site {\em unstable}. If all sites are stable, then we say that the configuration is stable.

Upon an {\em addition} at site $x$, the height of site $x$ increases by 1, and can henceforth become unstable. An unstable site may {\em topple}. With a toppling of site $x$, the height of site $x$ decreases by 2, and the height of the neighboring sites may increase, in such a way that the total increase is at most 2. We denote the neighbors of site $x$ as $x-1$ and $x+1$, where on the graph $C_n$, we interpret this modulo $n+1$. On $I_n$, 0 and $n$ only have 1 neighbor.
We consider three possibilities for the height increase of the neighbors in a toppling, namely

\begin{enumerate}
\item The height of site $x-1$ (if present) increases by 2, and the height of site $x+1$ does not change. We call
this a {\em left} toppling ($l$).
\item The height of site $x+1$ (if present) increases by 2, and the height of site $x-1$ does not change. We call
this a {\em right} toppling ($r$).
\item The height of both sites $x-1$ and $x+1$ (if present) increase by 1. We call this a {\em symmetric} toppling ($s$).
\end{enumerate}

In this paper, we introduce {\em quenched sandpile models}, which we define as follows. On a graph $G$, we let $\nu$ be a probability measure on $\{l,r,s\}^G$, an element of which we call a {\em toppling environment}. We choose a toppling environment $\zeta$ according to $\nu$. Given a toppling environment $\zeta$, we allow at site $x$ only topplings of type $\zeta(x)$. We say that $\zeta(x)$ is the {\em label} at site $x$. The combination of a configuration $\eta$ and a toppling environment $\zeta$ is called a {\em state}, and denoted $(\eta,\zeta)$. We call this state stable if $\eta$ is a stable configuration.

If a state is unstable then topplings can lead to new unstable sites. We say we {\em stabilize} a state $(\eta,\zeta)$ if we keep toppling unstable sites, according to the (fixed) toppling environment, until there are no unstable sites anymore. It may happen though that we never obtain a stable configuration this way; in this case we say that the state is not stabilizable. If a stable configuration can be reached, then the order in which we perform topplings has no effect on the final stable configuration; this is the famous Abelian property of the sandpile (see \cite{FMR} for some technical details on this fact in the case of infinite volume). When $l$- and $r$-labels are not allowed, then this model reduces to the classical sandpile model.

The {\em Manna model} is the annealed version of the above quenched sandpile model. One convenient way to describe the Manna model \cite{rolla}, is that instead of a single label, we have a stack of labels at every site. Every stack consists of an infinite sequence of i.i.d.\ random labels, which are $s$ with probability 1/2, $l$ with probability 1/4 and $r$ with probability 1/4, independent of all other stacks. When a site topples, it topples according to the first label in its stack, and then removes that label from the stack.

This paper is organized as follows: In Section \ref{maindensities}, we define and state our results on the two main critical densities, namely the stationary and the transition density. In Section \ref{howtosimulate}, we discuss the additional critical densities that have been introduced in the literature to simulate $\rho_{tr}$. We state a number of new results on these densities. Finally, we give the proofs of all our results in a separate Section \ref{proofs}.

\section{$\rho_s$ and $\rho_{tr}$ for the quenched sandpile models}
\label{maindensities}

\subsection{Driven sandpiles on $I_n$ and the stationary density}

Driven sandpile models are defined on $I_n$. Let $\mu$ be a probability measure on $\{0,1\}^{\N}$ and $\nu$ a translation-invariant probability measure on $\{l,r,s\}^{\N}$. We use $\mu$ to draw an initial configuration, and $\nu$ to draw a toppling environment. More precisely, to implement the process on $I_n$, the initial configuration $\eta_n^0$ is the restriction of a drawing from $\mu$ to the first $n+1$ coordinates, and likewise, the toppling environment $\zeta_n$ is the restriction to the first $n+1$ coordinates of a drawing from $\nu$.

Dynamics is as follows. Given $\eta^0$, we add a particle at a randomly chosen site (uniformly over all sites) and stabilize the resulting configuration. The result is denoted $\eta_n^1$. Repeating this process leads to to a sequence of configurations $(\eta_n^t)$, $t=0,1,\ldots$.

We are interested in the `average' number of occupied sites in the process on $I_n$ in the limit as $n \to \infty$. For given $n$ and $\zeta$, the process is a Markov chain, but it turns out that for certain toppling environments, this Markov chain is
periodic so that we do not have a well defined unique stationary distribution. Nevertheless, the `average' hinted at above often exists in an unambiguous way. We denote by $\rho_n^t=\rho_n^t(\mu,\nu)$ the expected fraction of sites in $I_n$ whose height is 1 at time $t$. Here, expectation is with respect to $\mu \times \nu$. We say that $\nu$ is {\em directed} if either $\nu(\zeta(0)=r)=0$ or $\nu(\zeta(0)=l)=0$.

\begin{theorem}
\label{frozdirstat}
Suppose that $\nu$ is directed. Let $f_r = \nu(\zeta(0)\neq s)$. Then
\begin{equation}
\label{bestaanlimiet}
\rho_s= \rho_s(\mu, \nu) :=\lim_{n \to \infty} \lim_{t \to \infty} \rho_n^t
\end{equation}
exists and is equal to
$$
1-\frac{f_r}{2}.
$$
In particular, $\rho_s$ does not depend on $\mu$. Note that in the classical case, $f_r=0$, and Theorem \ref{frozdirstat} reduces to the well known result which we described in the introduction.

\end{theorem}

The quantity $\rho_s$ is called the {\em stationary density}.

\subsection{Infinite volume sandpiles on $\Z$ and the transition density}

Infinite volume sandpiles in this article are defined on $\Z$. We choose the initial state $\eta$ according to $\mu_\rho$, a Poisson product measure with parameter $\rho$. The toppling environment $\zeta$ is chosen as before, according to a stationary and ergodic probability measure $\nu$ on
$\{l,s,r\}^{\Z}$. There are no additions in this model, only topplings. At any discrete time $t$, we simply topple all unstable sites once. This implies that for every $t$, the distribution of $\eta^t$ is translation-invariant.

We define the {\em transition density} as
\[
\rho_{tr} = \rho_{tr}(\nu)= \sup \{\rho: (\eta,\zeta) \mbox{ is stabilizable } (\mu_\rho \times \nu)-\mbox{ a.s.}\}.
\label{rhotr}
\]

As mentioned in the introduction, it was recently proved for the classical sandpile model on the bracelet graph \cite{FLW} (with proper interpretations of $\rho_s$ and $\rho_{tr}$) that $\rho_s \neq \rho_{tr}$. The current belief is that this inequality holds in general for classical sandpile models, and that the one-dimensional case is an exception.

In this paper, we prove that $\rho_s \neq \rho_{tr}$ for quenched sandpiles, as long as $\nu$ is directed and satisfies
\begin{equation}
\label{restriction}
\nu (\zeta: \zeta(x)=\zeta(x+1)=s) = 0 \mbox{ \rm and } \nu(\zeta(x)=\zeta(x+2)=s) = 0.
\end{equation}
for all $x$. In words, this condition means that between every pair of $s$ labels there must be at least two $r$ labels. Note that the transition density is not interesting for non-directed models: indeed, between $r$- and an $l$-labelled sites, particles cannot escape, and since it is possible with positive probability that too many particles are trapped in between, the model will not stabilize a.s., for any value of $\rho >0$.

\begin{theorem}
Suppose that $\nu$ is directed and satisfies (\ref{restriction}).
Let $f_r = \nu(\zeta(0)=r)$. Then we have
\begin{itemize}
\item[(a)] $\rho_s = 1-\frac{f_r}{2}$.
\item[(b)] $\rho_{tr}$ is the unique solution of $\rho-\frac 12(1 + e^{-2\rho}) =\frac {1-f_r}2 (1 + e^{-2\rho})^2$.
\end{itemize}
\label{frozrestricted}
In particular, $\rho_s \neq \rho_{tr}$ for $\nu$ satisfying (\ref{restriction}).
\end{theorem}

\section{Simulating $\rho_{tr}$: more critical densities}
\label{howtosimulate}

\subsection{Fixed-energy sandpile models and the threshold density}

The fixed-energy sandpile models in this article are defined on $C_n$. They play an important role in simulating $\rho_{tr}$. In the initial state, $\eta_n$ is chosen according to $\mu_{\rho}$, a Poisson product measure with parameter $\rho$, and $\zeta_n$ is chosen according to $\nu$.
The fixed-energy model evolves in time by topplings. If the initial state is stabilizable, then every site topples only finitely many times.

We define the {\em threshold density} as
\[
\rho_{th}=\rho_{th}(\nu) = \lim_{n \to \infty} \sup \{\rho: (\mu_{\rho} \times \nu )((\eta_n,\zeta_n) \mbox{ is stabilizable})\geq 1/2\}.
\label{rhoth}
\]
For the classical fixed-energy sandpile model, the threshold density is equal to 1.

Note that the definition of $\rho_{th}$ is very similar to that of $\rho_{tr}$. It is often tacitly assumed that they are equal for the classical sandpile on general graphs: in \cite{FLW}, this assumption is explicitly stated. In the literature, only for the classical sandpile in dimension 1 a proof is available. Note that for a quenched sandpile model, if both $l$-sites and $r$-sites exist, $\rho_{th}=0$, since it is always possible to have too many particles between a $r$-site and an $l$-site, from which particles can never escape. Hence, again, only directed models are of interest. We prove the following.

\begin{theorem}
\label{opnieuw}
Suppose that $\nu$ is directed and satisfies (\ref{restriction}).
Then $\rho_{tr}=\rho_{th}$.
\end{theorem}

\subsection{Parallel chip firing and the activity density}

Consider fixed-energy models with parallel toppling order. These cellular automata are called {\em parallel chip firing} models \cite{prisner}.
We are interested in the {\em activity density} $\rho_a(\mu_{\rho}, \nu, n)$, which loosely formulated, is the eventual average fraction of unstable sites when parallel chip firing is run on a finite graph, $C_n$ in our case. The next step is to take the limit as $n \to \infty$:
\[
\rho_a(\mu_{\rho}, \nu) =  \limsup_{n \to \infty} \rho_a(\mu_{\rho}, \nu, n).
\]
For the classical sandpile on $C_n$, the graph of $\rho_a$ versus $\rho$ has been shown \cite{asta} to take the shape of a `staircase' with three stairs (see Figure \ref{threestairsfig}). More specifically, the classical model on $C_n$ behaves in the following particularly simple way. If the total number of particles is strictly in $(n,2n)$, then every initial configuration will eventually have period 2. Thus, as $t \to \infty$, the total number of unstable sites in one time step will eventually be close to $n/2$. If the total number of particles is larger than $2n$, then eventually all sites will be unstable. Finally, it was already known that if the total number of particles is smaller than $n$, then every initial configuration stabilizes, so that eventually the total number of unstable sites is 0. Thus, the staircase of the classical sandpile model on $C_n$ has three stairs, and it does not even depend on the precise initial distribution of the particles.
Only if the total number of particles is equal to $n$ or $2n$, there are multiple possibilities for the limiting fraction of unstable sites.

\begin{figure}
\includegraphics[width=8cm]{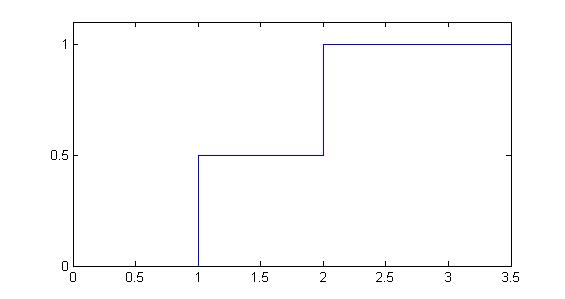}
\caption{The staircase graph for the classical sandpile model on $C_n$. The horizontal axis is $\rho$, the vertical axis is the activity density. The graph is the same for the fixed-energy and the infinite volume model; see \cite{asta} and Theorem
\ref{threestairs}.}
\label{threestairsfig}
\end{figure}

For the flower graph \cite{FLW}, the corresponding plot has five stairs.
In the case of the two-dimensional torus, it is conjectured, based on numerical observations \cite{bagnoli}, that the plot is a devil's staircase. Levine proved the occurrence of a devil's staircase for parallel chip firing on the complete graph, with a special class of initial configurations \cite{levine}.
To show that the plot need not always be a staircase, we include a simulation for the sandpile model of Theorem \ref{frozrestricted} (see Figure \ref{restrictedstaircase}).

\begin{figure}
\includegraphics[width=8cm]{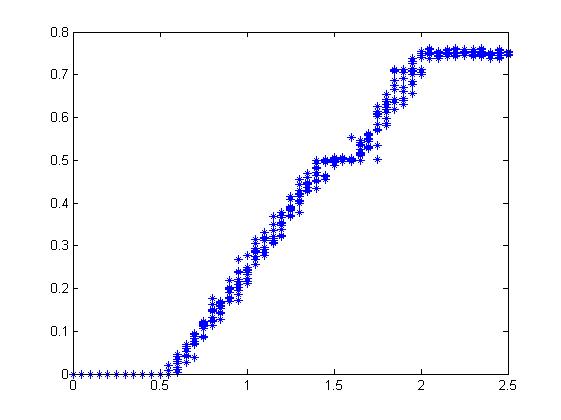}
\caption{Simulation of the the quenched restricted model on $I_{1000}$, with $f_s = 1/4$. There are 10 trials for each of 50 values of $\rho$. In each trial, a new random initial state was generated. Numerically, the predicted value for $\rho_{tr}$
(Theorem \ref{frozrestricted}) in this case is 0.556... }
\label{restrictedstaircase}
\end{figure}

The activity density has been used to simulate $\rho_{tr}$, under the assumption that $\rho_{tr}$ is equal to the so called {\em critical activity density} $\rho_{ca}$ defined as
\[
\rho_{ca} =\rho_{ca}(\nu):= \sup \{\rho: \rho_a(\mu_{\rho},\nu)=0\}.
\]
One can prove this for a number of specific graphs that have an infinite volume counterpart, for instance $C_n$ (where both densities are trivially equal to 1), or, less obvious, for the bracelet graph \cite{FLW}. The flower graph or the complete graph however do not have an infinite volume counterpart.

There is an infinite-volume analog of the activity density. Let $\alpha_{\rho}(n,t)$ be the fraction of active sites in $[-n,n]$ after $t$ iterations of the process, starting with $\mu_{\rho}$. For fixed $t$, the limit $\alpha_{\rho}(t):=\lim_{n \to \infty} \alpha_{\rho}(n,t)$ exists by ergodicity, and is a $\mu_{\rho}$-a.s.\ constant. The following result extends \cite{asta} to infinite volume.

\begin{theorem}
\label{threestairs}
In the infinite volume classical sandpile model on $\Z$ we have the following.
\begin{enumerate}
\item For $\rho <1$, $\lim_{t \to \infty} \alpha_{\rho}(t)=0.$
\item For $\rho >2$, $\lim_{t \to \infty} \alpha_{\rho}(t)=1.$
\item For $1 < \rho < 2$, $\lim_{T \to \infty} \frac{1}{T} \sum_{t=1}^T \alpha_{\rho}(t)= \frac{1}{2}.$
\end{enumerate}
\end{theorem}
In words, below $\rho=1$, the fraction of unstable sites converges to 0, above $\rho=2$ it converges to 1, and in between the fraction converges, in a weak sense, to $1/2$. The infinite-volume analog of the critical activity density is therefore equal to 1, and in the classical sandpile in one dimension, all critical densities are equal.

\section{Critical densities for the Manna model}
\label{mannaconjecture}

It is clear how the stationary density should be defined for the Manna model, and it is conjectured that it is strictly less than 1: Sadhu et al \cite{sadhu} claim that for the Manna model, $\rho_s = 0.953\ldots$, based on extrapolation of exact calculations for $n \leq 12$.
Theorem \ref{frozdirstat} perhaps supports the plausibility of the conjecture that $\rho_s < 1$ for the Manna model.

As far as $\rho_{tr}$ is concerned, Dickman et al.\ \cite{dickman} have estimated $\rho_{tr}$ for the Manna model as $0.94887$, in fact estimating the critical activity density on $C_n$. Huynh et al.\ \cite{moremannasims} performed extensive simulations for the Manna model on many one- and two-dimensional lattices. Their estimates confirm the result of \cite{dickman}, and in general give a value close to 0.9 for one-dimensional lattices, and close to 0.7 for two-dimensional lattices (note that the table with particle densities is not included in the first version of this paper).
Sidoravicius and Rolla \cite{rolla} proved that $\rho_{tr} \geq 1/4$. Sadhu et al.\ \cite{sadhu} compare their value for $\rho_s$ of the Manna model to the estimate for $\rho_{tr}$ of \cite{dickman}, implying that they should be equal (See also Conjecture \ref{mainconjecture} below).

It is not hard to show that for the Manna model, $\rho_{th} = 1$. For completeness, we state and prove this now.

\begin{proposition}
For the Manna model on $\Z$, $\rho_{th}=1$.
\label{obvious}
\end{proposition}

\begin{proof}
We make use of the following fact for the classical sandpile model on $C_n$: If the total number of particles is less than $n$, then the model will stabilize after finitely many topplings per site, irrespective of the initial position of the particles. This can be deduced in the following manner: if the model stabilizes, then at least one site does not topple. Then the neighbors of that site topple at most once, and so on.

Now for the Manna model on $C_n$, we choose $\rho < 1-\epsilon$, with $\epsilon>0$ arbitrarily small. Then the probability that the total number of particles is more than $n(1- \epsilon/2)$, tends to 0 as $n \to \infty$.
If the total number of particles is less than $n$, then the model stabilizes a.s. To see this, note that one possibility for stabilizing is if there is some time $t$ such that starting from $t$, all unstable sites topple according to an $s$ label until the model stabilizes. Since all labels are i.i.d. and only finitely many $s$ labels are needed, this event occurs a.s. Therefore, the Manna model on $C_n$ stabilizes with probability tending to 1 as $n \to \infty$ for all $\rho<1$, so that $\rho_{th} = 1$.
\end{proof}

We conclude that the critical activity density should be at least 1. However, the simulations for the Manna model arrive at a value strictly less than 1. The discrepancy results from a subtlety in choosing the observation time window. It is observed that there is a density value strictly smaller than 1, such that for $\rho$ below this value, the total number of topplings before stabilization is relatively small, whereas for $\rho$ above this value, the total number of topplings is much larger. It is this value that is taken to be an estimate for $\rho_{tr}$, rather than $\rho_{th}$ as defined in \eqref{rhoth}.
In the second case, before eventually stabilizing, the models is observed to settle into a `quasi-stationary' state, where the fraction of active sites is approximately constant for a long time; see \cite{slowrelax}, or Figure 3 of \cite{dantas} for numerical simulations of this state. It is believed that this quasi-stationary state of the fixed-energy Manna model reflects the true stationary state of the infinite volume Manna model. Therefore, the observation time window is chosen such that the model is observed in the quasi-stationary state, rather than large enough so that the model has stabilized.

Thus the findings of \cite{dickman} and \cite{dhar} indicate the following conjecture for the Manna model, which we now state explicitly. We stress that the conjectured relation between the three densities is markedly different from what is known for other models. For instance, for the classical sandpile model on the bracelet graph it is known \cite{FLW} that $\rho_s \neq \rho_{tr} = \rho_{th}$, and combining Theorems \ref{frozrestricted} and \ref{opnieuw} of this paper gives the same relation for the directed quenched sandpile models satisfying \eqref{restriction}.

\begin{conjecture}
For the Manna model on $\Z$,
\[
\rho_{s} = \rho_{tr} \neq \rho_{th} (=1).
\]
\label{mainconjecture}
\end{conjecture}

With kind permission of Lionel Levine, we give an example to illustrate quasi-stationarity in the Manna model. In this model, there is a value of $\rho$ such that below this value, the total number of topplings scales as $O(n \log n)$, whereas for $\rho$ above this value the total number of topplings increases exponentially in $n$.

\begin{example}
{\em
Consider the fixed-energy Manna model on the complete graph $K_n$, with self-loops. A site is active if there are at least two particles, and in a toppling each of the two particles chooses a new position uniformly at random from its neighbors, which in this case means from all sites.

We can view this sandpile model as an Ehrenfest urn model, with a stopping criterion. We say that the $n$ sites of $K_n$ represent $n$ balls in an urn, of two colors. One color (green) represents a site with an odd number of particles, and one color (red) a site with an even number of particles. In an Ehrenfest urn model, a move consists of picking one ball at random from the urn, removing it, and instead adding a ball of the opposite color.

Now consider a toppling of an active site in the sandpile model. We imagine that we first (temporarily) remove the two particles from the active site, and then one by one add them to their new positions. Removing the particles does not change any of the parities. Putting back one particle changes the parity of one site chosen uniformly at random from the graph, and likewise for the other particle.
Thus, a toppling of an active site is the equivalent of two moves in the urn model.

The sandpile model is stable if and only if each particle is on a different site. Therefore, if the particle number is $m$, then the stopping time for the according urn model is the first time that the number of green balls is $m$ (if $m>n$, then the model never stabilizes).

Now suppose that $n$ is large. Then typically, the number of particles is close to $n\rho$. In the initial configuration, typically the number of sites with an odd number of particles, is close to $\frac n2 (1-e^{-2\rho})$ ($n$ times the probability that a Poisson($\rho$) number is odd), which is strictly less than $n\rho$. In an Ehrenfest urn model, from any starting point the fraction of green balls typically tends to its equilibrium value of $1/2$.

Therefore, if $\rho < 1/2$, then the model meets the stopping criterion while tending to equilibrium. It is known that the expected total number of moves to reach equilibrium for the first time, starting from no green balls, is of the order $n \log n$ \cite{mahmoud}, only slightly more than linearly increasing with $n$. However, when $\rho>1/2$, then the model does not meet the stopping criterion while tending to equilibrium. On the contrary, it meets the stopping criterion only in a large deviation from equilibrium. This will typically take an exponentially large number of moves. Meanwhile, the number of green balls fluctuates round its equilibrium number, which we can describe as quasi-stationary behavior. Therefore, there is a sharp transition in the total number of topplings needed to stabilize, at $\rho=1/2$. If we monitor the quasi-stationary behavior, then we find 1/2 as an estimate for $\rho_{tr}$. Actually, the threshold density is $\rho_{th}=1$.

To find the stationary density, we say that one site is the sink. Then both topplings and additions are equivalent to one or more moves in the Ehrenfest urn model, this time without stopping criterion. Therefore, $\rho_s = 1/2$. We conclude that this model behaves as conjectured for the Manna model on $\Z$.
}
\end{example}

\section{Proofs}
\label{proofs}

\subsection{Proof of Theorem \ref{frozdirstat}}

Our main tool for the proof of Theorem \ref{frozdirstat} the following lemma. Consider the set $A=\{\eta : \eta(x)=1 \mbox{ for all } 0 \leq x <n\}$, which contains two elements which we denote by $\eta_f=(1,1,\ldots, 1)$ and $\eta_e=(1,1,\ldots, 1,0)$.

\begin{lemma}
Consider the driven sandpile model on $I_n$ with toppling environment $\zeta$ given by $\zeta(0)={r}$ and $\zeta(x)={s}$
for all other $x \in I_n$. Then the set $A$ defined above is absorbing. More precisely, denoting the addition site by $a$, transitions between $\eta_f$ and $\eta_e$ are as follows. If $n-a$ is even, then the configuration flips from $\eta_f$ to $\eta_e$ or vice versa; if $n-a$ is odd, then the configuration does not change upon addition at $a$ and subsequent stabilization.
\label{blockboringstuff}
\end{lemma}

Since obviously, $A$ is reachable from any configuration, the lemma implies that there are only two recurrent configurations: the `full configuration' $\eta_f$ in which all sites have height 1, or the `non-full configuration' $\eta_e$ in which all sites have height 1, except site $n$.

\begin{proof}
We use the well known description \cite{redig} which gives the result of one addition and subsequent stabilization in a particle configuration $\eta$ in the driven model on $I_n$: Let $a$ be the addition site. Let $i$ be $\max \{-1, i\leq a : \eta (i)=0\}$, that is, $i$ is the closest empty site to the left of $a$, and if there is no such site, then $i = -1$. Similarly, let $j$ be
$\min \{n+1, j \geq a: \eta (j)=0\}$. The resulting configuration after the addition and subsequent stabilization is denoted $\xi$. Then only sites in $(i,j)$ have toppled, and $\xi$ is as follows: sites $i$ and $j$ have gained 1 particle and site $i+j-a$ has lost 1 particle.
Note that if $\eta(a)=0$, then the above description simply states that site $a$ will become full. If $\eta(a)=1$, then at most three sites change height.

Now we return to the sandpile model with $\zeta(0)={r}$, and let $\eta \in A$. We will show that if $n-a$ is odd, then we end up in $\eta$ after addition and subsequent stabilization. If $n-a$ is even however, then we flip to the other element in $A$.

We first note that if the particle configuration is non-full and $a=n$, then it is obvious that the result is the full configuration.
Now suppose that $a$ is a full site. After the addition, we stabilize in steps. In the first step, we perform only $s$-topplings (that is, we do not topple site 0) until all sites except site 0 are stable. This step is equivalent to stabilization in the driven classical sandpile on $[1,n]$, with sinks at sites 0 and $n+1$.
In the second step, we topple site 0 once, and subsequently again only perform $s$-topplings until all sites except site 0 are stable. We repeat the second step until the whole configuration is stable.

To find the particle configuration after the first step, we apply the above description with $i=0$ and $j = j_1$ where $j_1=n$ if $\eta = \eta_e$, and $j_1=n+1$ if $\eta = \eta_f$. Then after the first step, site 0 now has two particles, and all other sites are full except for site $j_1-a$, which is now empty.

To find the particle configuration after the second step, we apply the above description twice, because toppling site 0 is equivalent to emptying site 0 and making two additions to site 1. Therefore, after the second step site 0 has again two particles, and all other sites are full except site $j_1-a-2$, which is now empty.

We repeat the second step $k$ times, with $k$ such that either $j_1-a-2k = 0$ or $j_1-a-2k = 1$. In the first case, we have reached a stable configuration, namely $\eta_f$. This will be the case if $j_1=n$ and $n-a$ is even, or $j_1=n+1$ and $n-a$ is odd. In the second case, all sites will topple one more time, so that the final configuration is $\eta_e$. This will be the case if $j_1=n+1$ and $n-a$ is even, or $j_1=n$ and $n-a$ is odd. 
This finishes the proof.
\end{proof}

The idea of the proof of Theorem \ref{frozdirstat} in words is as follows: we think of almost the entire state of being built up of `blocks' in which the leftmost site is ${r}$, and all other sites are ${s}$. Each of these blocks is like the model in Lemma \ref{blockboringstuff}: after a transient period, either all sites are full, or only the rightmost site is empty. Moreover, we will find that at stationarity, these configurations are equally likely. Therefore, since the proportion of ${r}$ sites is $f_r$, the proportion of empty sites at stationarity is $f_r/2$.

\begin{proof}[Proof of Theorem \ref{frozdirstat}]
Let $(\eta, \zeta)$ be a state on $I_n$ of the quenched directed model. Let $R_n$ be the number of ${r}$ sites in $\zeta_n$,
and call $X_i$, $i = 1, \ldots, R_n$ the positions of the ${r}$ sites, in increasing order.
We define $B_0=[0, X_1-1]$ and $B_i = [X_i, X_{i+1}-1]$ for $i=1,\ldots, R_n$. We refer to the $B_i$'s as a `block' and denote by $l_i$ be the number of sites in $B_i$.

The block $B_0$ behaves as a classical sandpile, implying that the only recurrent (sub) configurations on $B_0$ are those with at most one site with height 0. For all other blocks, Lemma \ref{blockboringstuff} implies that there are two recurrent sub-configurations, namely all $l_i$ sites are full (a `full' block), or all sites are full except the rightmost site $X_{i+1}-1$ (a `non-full' block).
In fact, according to Lemma \ref{blockboringstuff}, an addition to the $a$-th site (counting from the left) in $B_i$ with $l_i-a$ odd will leave the configuration unchanged, whereas an addition to the $a$-th site in $B_i$ with $l_i-a$ even will flip the block, that is, change its configuration from full to non-full or vice versa.

After a random but a.s.\ finite time, all blocks $B_1, \ldots, B_{R_n}$ are either full or non-full, and in $B_0$ there is at most one empty site. We will now argue that for each block $B_1, \ldots, B_{R_n}$, the probability that it is full converges to $1/2$ as time tends to infinity.

Consider the effect of one addition. Suppose we add at a site in $B_i$ for some $i \geq 1$. As noted before, either $B_i$ flips, or it does not, depending only on the location of the addition and its relative position in $B_i$. In particular, whether or not the block flips does not depend on the configuration. If $B_i$ flips, then either zero or two particles leave $B_i$. They have the effect of an even number of additions to site $X_{i+1}$. Since one addition at $X_{i+1}$ has either no effect, or results in a flip, two additions have the effect of no change at all. Therefore, the configuration of block $B_{i+1}$ does not change, and as a result, two or zero particles leave $B_{i+1}$ into $B_{i+2}$. Continuing this way, we conclude that if $B_i$ flips, then no other block flips.
If $B_i$ does not flip, then one particle leaves $B_i$, and it has the effect of one addition to site $X_{i+1}$. Therefore, $B_{i+1}$ will flip if $l_i$ is odd, otherwise, one particle leaves $B_{i+1}$ and has the effect of one addition to site $X_{i+2}$. We conclude that if $B_i$ does not flip, then `the first odd block downstream' flips. In other words, block $B_j$ flips, where $l_j$ is odd, and there is no $i<k<j$ such that $l_k$ is odd. If there exists no such $j$, then no block flips, so that the addition did not change the configuration.

In the interval $[0,X_1-1]$ there is at most one empty site, but it is not necessarily site $X_1 -1$. If the addition site is in $B_0$, then the new configuration depends on the current configuration. However, at most two particles leave this interval: either one site flips from 0 to 1, one site flips from 1 to 0, or two sites flip (one from 0 to 1 and one vice versa). In the first two cases, zero or two particles leave $B_0$, so that the configuration downstream is not affected. In the third case, one particle leaves $B_0$, so that the first odd block downstream flips.

Thus we can partition all sites, except the sites in $[0,X_1-1]$, into deterministic non-empty sets $S_i$, $i=1,\ldots, R_n$, such that if there is an addition to a site in $S_i$, then $B_i$ and only $B_i$ flips. We can summarize the effect of an addition as follows: pick a random block, such that for all $i=1,\ldots, R_n$, block $B_i$ is chosen with probability $\frac{|S_i|}{n+1}$, and flip it. With probability $\frac{X_1}{n+1}$, flip either a site in $[0,X_1-1]$, or two in $[0,X_1-1]$ and the first odd block. With probability $\frac{(n+1 - X_1 - \sum_i |S_i|)}{n+1}$, flip nothing.

It is clear from the above description that for each block $B_1, B_2, \ldots, B_{R_n}$, the probability that it has flipped an even number of times until time $t$ (inclusive, say) converges to $1/2$ as $t \to \infty$. Hence, the expected number of sites with height 0 in block $B_i$, converges to $1/2$ as $t \to \infty$, for all $i=1, \ldots, R_n$. Block $B_0$ can be viewed as a classical sandpile model on $X_1$ sites (with a slower addition rate of course). Therefore, the expected number of sites in $B_0$ with height 0 converges to $1/(x+1)$ as $t \to \infty$.

Writing $\prob$ for $\mu \times \nu$ and $\expec$ for the corresponding expectation, we now write
\begin{eqnarray*}
\rho_n^t &=& \frac{1}{n+1}\sum_{k,x}\expec \left(\sum_{i=0}^n \eta^t(i)|  R_n=k,X_1 = x\right)\prob(R_n=k,X_1=x),
\end{eqnarray*}
which is a finite sum.
When $t \to \infty$, this converges to
\begin{eqnarray*}
\lim_{t \to \infty} \rho_n^t & = & \sum_{k,x}\left( 1-\frac{k}{2(n+1)} -\frac{1}{(x+1)(n+1)}\right) \prob(R_n=k, X_1=x)\\
&=& 1-\frac{1}{2(n+1)}\expec(R_n)-\sum_{k,x} \frac{1}{(x+1)(n+1)} \prob(R_n=k, X_1=x).
\end{eqnarray*}
Since $\frac 1{n+1} \expec (R_n)$ tends to $\nu(\zeta: \zeta(0)=r) = f_r$ as $n \to \infty$ and the last term tends to 0, we obtain
\begin{eqnarray*}
\lim_{n \to \infty} \lim_{t \to \infty} \rho_n^t &=& 1 - f_r/2.
\end{eqnarray*}
\end{proof}

\subsection{Proof of Theorem \ref{frozrestricted} and Theorem \ref{opnieuw}}

\noindent
{\em Proof of Theorem \ref{frozrestricted}.} Part (a) follows from Theorem \ref{frozdirstat}, hence we need only to prove (b). Recall that the initial configuration $\eta$ consists of a Poisson-$\rho$ number of particles at every site. It will be convenient to also consider the parity of the particle numbers at each site. To this end, let $\eta' = \eta \mod 2$,
and call $\eta'(x)$  the {\em parity} of $x$.
We define $E(x) = (\eta(x)-\eta'(x))/2$, and call $E(x)$ the number of {\em excess pairs of particles} at $x$. The $E(x)$ will play a crucial role in what follows. When a site $x$ is unstable, it topples, which means that it loses two particles. If $\zeta(x)=r$, the two particles can be thought of as staying together, while when $\zeta(x)=s$, the two particles are separated, since they are sent to different neighbors. It is for this reason that we need to analyze in detail how topplings of an $s$-site take place, since only around such $s$-sites it is possible that the excess pairs can be absorbed.

Consider a site $x$ with $\zeta(x)=s$. For the time being, we restrict our attention to $\{x-1, x, x+1\}$. We topple unstable sites in $\{x-1, x, x+1\}$, until we reach the final stable configuration $\xi$ on $\{x-1, x, x+1\}$. Since we work in the restricted model, we have that $\zeta(x-1)=\zeta(x+1)=r$. First, observe that topplings will never change the parity of $x$.

Furthermore, the sum of the parities of $x-1$ and $x+1$, computed $\mod 2$, does not change either under topplings.
Finally, note that if $x$ topples at all, then $\xi(x-1)$ can not be 0, and hence $\xi(x-1)=1$. Hence, we have the following table, dealing with the parities before versus parity/configuration after stabilization of $\{x-1, x, x+1\}$, under the assumption that $x$ topples at least once.
\begin{center}
\begin{tabular}{|c|c|}
\hline
 parity before stabilization & configuration after stabilization \\
 \hline
 000 & 101 \\
 \hline
 001 & 100 \\
\hline
 010 & 111 \\
 \hline
 100 & 100 \\
 \hline
 011 & 110 \\
 \hline
 101 & 101 \\
 \hline
 110 & 110 \\
 \hline
 111 & 111 \\
 \hline
\end{tabular}
\end{center}
It follows from the table that the set $\{x-1, x, x+1\}$ can absorb an excess pair if and only if the original parities are 000 or 010. We call a site $x$ with $\zeta(x)=s$ a {\em hole}, if $\eta'(x-1)=\eta'(x+1)=0$. In all other cases, all excess pairs that are present in, or enter $\{x-1, x, x+1\}$ from the left, will after some topplings in $\{x-1, x, x+1\}$ result in excess pairs leaving at the right of the set.
For every excess pair, this requires at most two topplings per site.
This observation suggests that in order to decide whether or not a configuration is stabilizable, we need to compare the spatial density of excess pairs to the density of holes. Or, in other words, we need to compare the probability that a given site $x$ is a hole with the expected number of excess pairs at $x$. Below, we formally prove that this is the correct picture, but first we compute these quantities.

It is easy to compute the (conditional) probability for a site $x$ to be a hole, given that $\zeta(x)=s$. Indeed, since the original configuration is product measure with Poisson-$\rho$ marginals, the probability that a site has parity 0 is
$\frac12(1 + e^{-2\rho})$. It follows that the probability of a site $x$ being a hole is
\begin{equation}
\label{nr1}
\frac{1-f_r}{4}(1+e^{-2\rho})^2.
\end{equation}
A similar and related computation shows that the expected number of excess pairs is equal to
\begin{equation}
\frac12(\rho-\frac12(1+e^{-2\rho})).
\label{nr2}
\end{equation}
Comparison of (\ref{nr1}) and (\ref{nr2}) then finishes the proof.

It remains to show that it is indeed the case that the relation between the probability of a hole and the expected number of excess pairs determines whether or not stabilization occurs. To this end, we first assume that (\ref{nr1}) $<$ (\ref{nr2}), that is, the spatial density of holes is smaller than the expected number of excess pairs per site. Suppose that stabilization occurs under this assumption.
First we note that if the configuration stabilizes, all excess pairs must be absorbed by a hole. Indeed, if there is an excess pair that is never absorbed, with positive probability, then by ergodicity there must be infinitely many such pairs to the left of the origin, which implies that the origin will topple infinitely many times. However, a pairing in which every excess pair is paired to a unique hole is only possible if the spatial density of the holes is at least as large as the spatial density of the excess pairs; this follows, for instance, from the well known mass transport principle, see e.g. \cite{schramm}. This violates the assumption, and we are done.

For the other direction, we assume that (\ref{nr1}) $>$ (\ref{nr2}). Define a random walk $S(x), x \in \Z$ as follows. First we put $S_0=0$.
Writing $h(x)=1$ if there is a hole at $x$, and $h(x)=0$ otherwise, we put
$$
S(x)= S(x-1) + E(x)-h(x).
$$
By assumption, $S$ has negative drift. We say that $x$ is {\em extremal} if $S(y)> S(x)$ for all $y <x$. The collection of extremal points has a well defined spatial density. In particular, there are infinitely many extremal points in both directions. It is now simple to see that any excess pair will be absorbed before reaching the first extremal point to the right. In particular, no excess pair will travel across an extremal point. From this it follows that the configuration stabilizes.
\qed

\medskip\noindent
{\em Proof of Theorem \ref{opnieuw}.}
The proof of this result is very similar to the proof of Theorem \ref{frozrestricted}. Again, we have to compare the number of excess pairs to the number of holes. For the same reason as before, each hole can accommodate one excess pair. What follows from this is that $\rho_{th}$ has the same numerical value as $\rho_{tr}$.
\qed

\subsection{Proof of Theorem \ref{threestairs}}

The main difficulty of the proof is that we work in infinite volume, so we cannot use as a starting point that the dynamics becomes periodic.

The cases $\rho <1$ and $\rho >2$ are not difficult.
The case $1 < \rho <2$ however is rather tedious.
Our strategy is to generalize the well-known notion of a {\em forbidden sub-configuration} (FSC) \cite{dhar, MRZ}, special local configurations which do not appear in the stationary state of the system. We will show that in infinite volume with parallel toppling order, there are a number of additional local configurations which, essentially, do not appear in the stationary state of the infinite-volume system in the sense that their spatial density tends to 0 as $t \to \infty$. Once we have proved that, it is easy to show that in all remaining possibilities, all sites have period 2, and this will conclude the proof.

In the proof, the following lemma plays an important role. In \cite{bagnoli} it was stated and proved for $\Z^2$, but their proof works for general graphs. We state a version suitable for our purposes. We say that two configurations $\eta$ and $\xi$ are {\em mirror images} of each other, if $\eta(x) = 3 -\xi(x)$ for all $x$.

\begin{lemma}[\cite{bagnoli}]
Let two height configurations $\eta$ and $\xi$ be mirror images of each other. Then after performing, for each, a parallel toppling time step, the two resulting configurations are again mirror images of each other.
\label{mirror}
\end{lemma}

\medskip\noindent
{\em Proof of Theorem \ref{threestairs}.}
The case $\rho < 1$ follows from the fact that the system is (a.s.) stabilizable when $\rho<1$. This is a well-known fact, proved in several papers, see for example \cite{quant}. We remark that in \cite{MRZ}, it is moreover proved that for $\rho=1$ the system is not stabilizable.

Next, consider the case $\rho>2$. We have that $\expec_\rho (3-\eta(0)) < 1$. Therefore, when $\eta$ is chosen according to $\mu_{\rho}$, the mirror image $\bar{\eta}$ of $\eta$ is a.s.\ stabilizable. So when we start the dynamics from both $\eta$ and $\bar{\eta}$, then for every site $x$, $\bar{\eta}^t(x) \in \{0,1\}$ eventually. According to Lemma \ref{mirror}, $\eta^t(x) \in \{2,3\}$ eventually, proving this case. (Note that when a site $x$, and its two neighbours, are all unstable, then after toppling all unstable sites, the height of $x$ is the same as before. So it is possible to have a well defined limiting configuration, despite the fact that topplings continue to take place.)

Suppose now that $1<\rho<2$. In this case, when we choose $\eta$ according to $\mu_{\rho}$, neither $\eta$ nor $\bar{\eta}$ is stabilizable a.s. Therefore, again using Lemma \ref{mirror}, for every site $x$, there is a (random) time $T< \infty$ such that if $t>T$, $x$ has been stable for at least one time step, and also has been unstable for at least one time step. Hence, for any finite subset $W$ of $\Z$, there is a (random) time
$T_W< \infty$ such that before time $T_W$, every site in $W$ has been both stabel and unstable for at least one time step. It is well known (see \cite{dhar} and also \cite{MRZ}) that this implies that for $t > T_W$, the total number of particles in $W$ is at least equal to the number of internal edges in $W$. A local configuration in $W$ which does {\em not} satisfy this property is called a {\em forbidden sub-configuration} (FSC). It also follows from \cite{dhar} that both $\eta_t(x)$ and $\bar{\eta}_t(x)$ are at most 3, for all $x \in W$ and $t > T_W$.

Comparing the dynamics from both $\eta$ and its mirror $\bar{\eta}$, we have from Lemma \ref{mirror} that if $\eta^t(x)$ is unstable, then $\bar{\eta}^t(x)$ is stable, and vice versa. Since $\bar{\eta}$ eventually has no FSC in $W$,
we see that in $\eta$, also mirrors of FSC's no longer appear eventually. For instance, since
$00$ is a FSC, and therefore does not appear anymore eventually, neither does e.g.\ $33$ occur in the long run, on any given position, and similarly for all other FSC's.

Apart from the absence of FSC's, we claim that also the patterns $01, 10$ and $111$ are very rare, in the sense that the spatial density of these finite-dimensional patterns converges to 0 as $t$ tends to infinity. Using Lemma \ref{mirror} or the above this implies that also $23, 32$ and $222$ have small spatial density for $t$ large.
We prove this claim at the end of this proof, but first show that it implies the desired result.

It is easy to check from the absence of the above mentioned local configurations, that apart from local perturbations which occur with spatial density tending to 0, the only possible configurations are those in which we alternatingly choose elements from $\{0,1,11\}$ and $\{2,3,22\}$. Any finite-dimensional configuration which is built up from alternatingly choosing elements from $\{0,1,11\}$ and $\{2,3,22\}$ has period 2. If the full configuration would be like that (that is, if the FSC's, $01, 10, 111, 222, 23$ and $32$ would not occur at all) than it would follow that
\begin{equation}
\label{somis1}
\alpha_{\rho}(t) + \alpha_{\rho}(t+1)=1,
\end{equation}
for all $t$. Indeed, if that were the case, then the configuration would be periodic with period 2, with each site alternating between a stable and an unstable state. However, we have to deal with the rare perturbations of the FSC's and the configurations $01, 10, 111, 222, 23$ and $32$. This however is not a real complication. For any $\epsilon >0$, we find $T$ so large that the sum of the spatial densities of all the problematic sub-configurations is at most $\epsilon$, for all $t >T$. It is clear that for these values of $t$, the sum in (\ref{somis1}) is close to 1. This is enough to prove the result.

It remains to prove the claim that the patterns 01, 10 and 111 have spatial densities which converges to 0 as $t \to \infty$. We start with 01 and 10. We write $\mu_{\rho}^t$ for the probability measure describing the state of the process at time $t$, hence $\mu_{\rho}^0=\mu_{\rho}$. Since the initial measure $\mu_{\rho}^0$ is ergodic, and $\mu^t_{\rho}$ is a factor of $\mu^0_{\rho}$, it follows that also $\mu^t_{\rho}$ is ergodic. Hence both 01 and 10 have a well-defined spatial density at time $t$, which is an a.s.\ constant, depending on $t$. By symmetry, these two densities must be the same. (We do not need that they are the same, but we do need that if one is positive, so is the other.)

Now let, as before, $W$ be a finite interval, and suppose that for some $t$, the configuration in $W$ contains the patterns 01 and 10 (in that order). When $\eta^t(x)=0$ and $\eta^t(x+1)=1$, we will say that the pattern 01 occurs {\em from} $x$ (the leftmost vertex of the pattern) and similarly for other local patterns.

Suppose 01 occurs from $x$ at a certain time $s$, where $s$ is much larger than $T_W$, how large exactly will be decided later. We are interested in the last time before time $s$ that the configuration at $(x, x+1)$ was different from 01; we call this the {\em previous} configuration. So the previous configuration is not necessarily the configuration at time $s-1$. Instead, it is the last configuration at $(x, x+1)$ which was different from the current one. The 0 at $x$ can only have been obtained through a toppling at $x$, and since $x+1$ has only one particle, site $x+2$ must have been stable. Hence the previous configuration from $x$ must have been either 200 or 201. Since 00 is an FSC, assuming that the previous configuration is at a time still larger than $T_W$, 00 is not allowed, so 201 is the only option. This means that the pattern 01 cannot be created after time $T_W$: if it occurs from $x$, then the previous configuration from $x$ must have been 201, implying that 01 occurs from $x+1$. A similar (symmetric) argument is valid for the occurrence of 10.

It follows that when we go backwards in time, as long as we remain larger than $T_W$, the pattern 01 `travels' to the right, and the pattern 10 `travels' to the left. This implies that at some time in the past they must have met, and we must have either the pattern 0110 or the pattern $01*10$ in $W$, where $*$ is arbitrary. All this, again, is only valid if we assume that when we go backwards in time, we never get to a time smaller than $T$ before one of the two patterns 0110 or $01*10$ occurs. It is clear that we can take $s$ so large that this is the case with overwhelming probability, since the model does not stabilize.

But now observe that 0110 is an FSC, and therefore this is not a possibility, since FSC's do not occur in $W$ after time $T_W$. Furthermore, the previous configuration of $01*10$ is an FSC whenever the $*$ is 0 or 1, and therefore these are not possible. If $*=2$, then the 2 must have arisen from the previous configuration through a toppling of (at least) one of its neighbours, but then the 0 next to the toppled site is impossible. Similarly for the possibility that $*=3$. We conclude that all cases lead to impossibilities, and therefore, on any finite interval $W$, we cannot have the patterns 01 and 10 together, after sufficiently long time. This implies that the (common) density of 01 and 10 must tend to zero, as time tends to infinity.

Next, we focus on the pattern 111, and again we ask the question what the previous configuration could have been. (Recall that the previous configuration is not the configuration one time unit earlier, but the last configuration on that position that was different from 111.) If the previous configuration contains a 0, then exactly one of the neighbours of this 0 must be unstable. Indeed, after the toppling of all unstable sites, this is the only way to go from a 0 to a 1. But this implies that the other neighbour of the 0 is a 0 or a 1, which means that the previous configuration contains either an FSC, or the configuration 01 (or 10), which we know by the previous paragraph has density tending to 0.

If the previous configuration contains a 1, then the neighbours of this 1 can not be unstable, and by the same argument as above, the neighbours of this 1 cannot be 0 or 1 (after time $T_W$). Hence, the only possibility is that the previous configuration would be 111, but this is a contradiction, since by definition, the previous configuration must be different from the current one. It remains to rule out all previous configurations which consist of only 2's and 3's. It is obvious however that it is not possible to reach 111 through such a previous configuration: if all sites are unstable then the central site receives two particles and is again unstable.

We summarize that if the pattern 111 occurs from $x$ at a given time larger than $T_W$, then either it was present from $x$ at time $T_W$ or one of the patterns 01 or 10 was present in the previous configuration. The first possibility requires that the central site has not toppled since time $T_W$, which has probability tending to $0$ as time tends to infinity. The second possibility gives us that the spatial density of the pattern 111 tends to 0, just as we showed for the patterns 01 and 10.
\qed

\paragraph{Acknowledgment:} We thank Lionel Levine, David Wilson and Alexander Holroyd for stimulating discussions.

\end{document}